\documentclass[a4paper,10pt,leqno]{article}
\usepackage{amsmath,amsfonts,amssymb,amsthm}
\usepackage{graphicx,subfigure,booktabs,multirow}
\usepackage{mathrsfs,enumerate,color,bm}
\usepackage{algorithm,algpseudocode}
\usepackage{charter}
\usepackage{geometry}
 \usepackage[colorlinks, linkcolor=red, citecolor=blue]{hyperref}
\newtheorem{theorem}{Theorem}[section]
\newtheorem{lemma}[theorem]{Lemma}
\newtheorem{proposition}[theorem]{Proposition}
\newtheorem{corollary}[theorem]{Corollary}

\newtheorem{example}{Example}[section]

\title{Elasticity $\mathscr{M}$-tensors and \\ the Strong Ellipticity Condition}
\author{
Weiyang Ding\thanks{Department of Mathematics,  Hong Kong Baptist University, Kowloon Tong, Hong Kong. Email: {\tt wyding@hkbu.edu.hk}. This author's work was partially supported by the National Natural Science Foundation of China(11801479).}
\and
Jinjie Liu\thanks{Department of Applied Mathematics, The Hong Kong Polytechnic University, Kowloon, Hong Kong. Email: {\tt jinjie.liu@connect.polyu.hk}.}
\and
Liqun Qi\thanks{Department of Applied Mathematics, The Hong Kong Polytechnic University, Kowloon, Hong Kong. Email: {\tt liqun.qi@polyu.edu.hk}. This author's work was partially supported by the Hong Kong Research Grants Council (Grant No. PolyU 15302114, 15300715, 15301716, 15300717 and C1007-15G).}
\and
Hong Yan\thanks{Department of Electronic Engineering, City University of Hong Kong, Kowloon, Hong Kong. Email: {\tt h.yan@cityu.edu.hk}. This author's work was partially supported by the Hong Kong Research Grants Council (Grant No. C1007-15G).}
}
\date{\today}

\begin{document}

\maketitle

\begin{abstract}
In this paper, we establish two sufficient conditions for the strong ellipticity of any fourth-order elasticity tensor and investigate a class of tensors satisfying the strong ellipticity condition, the elasticity $\mathscr{M}$-tensor. The first sufficient condition is that the strong ellipticity holds if the unfolding matrix of this fourth-order elasticity tensor can be modified into a positive definite one by preserving the summations of some corresponding entries. Second, an alternating projection algorithm is proposed to verify whether an elasticity tensor satisfies the first condition or not. Besides, the elasticity $\mathscr{M}$-tensor is defined with respect to the M-eigenvalues of elasticity tensors. We prove that any nonsingular elasticity $\mathscr{M}$-tensor satisfies the strong ellipticity condition by employing a Perron-Frobenius-type theorem for M-spectral radii of nonnegative elasticity tensors. Other equivalent definitions of nonsingular elasticity $\mathscr{M}$-tensors are also established.



\end{abstract}

\bigskip

\section{Introduction}

The strong ellipticity condition is essential in theory of elasticity, which guarantees the existence of solutions of basic boundary-value problems of elastostatics and thus ensures an elastic material to satisfy some mechanical properties.
Thus to identify whether the strong ellipticity holds or not for a given material is an important problem in mechanics \cite{Gurtin73}.
Knowles and Sternberg \cite{KnowlesSternberg75,KnowlesSternberg76} proposed necessary and sufficient conditions for strong ellipticity of the equations governing finite plane equilibrium deformations of a compressible hyperelastic solid.
Their works were further extended by Simpson and Spector \cite{SimpsonSpector} to the special case using the representation theorem for copositive matrices.
Rosakis \cite{Rosakis90} and Wang and Aron \cite{WangAron96} also established some reformulations.
Furthermore, Walton and Wilber \cite{WaltonWilber} provided sufficient conditions for strong ellipticity of a general class of anisotropic hyperelastic materials, which require the first partial derivatives of the reduced-stored energy function to satisfy several simple inequalities and the second partial derivatives to satisfy a convexity condition.
Chiri\c t\u a, Danescu, and Ciarletta\cite{ChiritaDanescuCiarletta07} and Zubov and Rudev \cite{ZubovRudev16} gave sufficient and necessary conditions for the strong ellipticity of certain classes of anisotropic linearly elastic materials.
Gourgiotis and Bigoni \cite{GourgiotisBigoni16} investigated the strong ellipticity of materials with extreme mechanical anisotropy.

Qi, Dai, and Han \cite{QiDaiHan09}  proved a necessary and sufficient condition of the strong ellipticity by introducing M-eigenvalues for ellipticity tensors and showing that the strong ellipticity holds if and only if all the M-eigenvalues of the ellipticity tensor is positive.
A practical power method for computing the largest M-eigenvalue of any ellipticity tensor was proposed by Wang, Qi, and Zhang \cite{WangQiZhang09} and may also be applied to the verification of the strong ellipticity.
Very recently, Huang and Qi \cite{HuangQi17} generalized the M-eigenvalues of fourth-order ellipticity tensors and related algorithms to higher order cases.
Another type of ``eigenvalues'' for ellipticity tensors called singular values was defined by Chang, Qi, and Zhou \cite{ChangQiZhou10}, and the positivity of all the singular values of the ellipticity tensor is also a necessary and sufficient condition for the strong ellipticity.
Han, Dai, and Qi \cite{HanDaiQi09} linked the strong ellipticity condition to the rank-one positive definiteness of three second-order tensors, three fourth-order tensors, and a sixth-order tensor.

Symmetric ${\bf M}$-matrices, also called the Stieltjes matrices, are an important class of positive semidefinite matrices used in many disciplines in science and engineering,  such as linear systems of equations, numerical solutions of partial differential equations, the Markov chains, the queueing theory, and the graph theory \cite{BermanPlemmons94}.
Zhang, Qi, and Zhou \cite{ZhangQiZhou14} introduced higher order $\mathscr{M}$-tensors and showed that an even-order symmetric $\mathscr{M}$-tensor is positive semidefinite.
Ding, Qi, and Wei \cite{DingQiWei13} proposed several equivalent definitions of nonsingular $\mathscr{M}$-tensors.
Note that the nonsingular $\mathscr{M}$-tensor is also called the strong $\mathscr{M}$-tensor in some literature.
Ding and Wei \cite{DingWei16} proved that there exists a unique positive solution of any polynomial system of equations whose coefficient tensor is a nonsingular $\mathscr{M}$-tensor and the right-hand side is a positive vector. They proposed an iterative algorithm for solving such systems.
Furthermore, Han \cite{Han2017}, Xie, Jin, and Wei \cite{XieJinWei17} and Li, Xie, and Xu \cite{LiXieXu17} also proposed other numerical methods. Very recently, Bai, He, Ling, and Zhou \cite{BHLZ2018} and Li, Guan, and Wang \cite{LGW2018} concerned the nonnegative solutions for the M-tensor equations.

Actually, the above $\mathscr{M}$-structure is defined with respect to the tensor eigenvalues introduced by Qi \cite{Qi05}.
According to \cite[Chapter 2]{QiLuo17}, a tensor whose M-eigenvalues are all positive (or nonnegative) is said to be M-positive (semi)definite.
We shall define the $\mathscr{M}$-tensors with respect to the M-eigenvalues, which will be shown to be M-positive semidefinite.
Subsequently, we can find a large class of tensors satisfying the strong ellipticity condition.

The rest of the paper is organized as follows.
We briefly introduce the strong ellipticity condition and its relationship with several types of positive definiteness in Section 2.
In Section 3,  a sufficient condition is given to verify the strong ellipticity of an elasticity tensor with an alternating projection algorithm.
Next, we investigate the M-spectral radius of nonnegative elasticity tensors as preparation for defining the elasticity $\mathscr{M}$-tensors in Section 4.
Then we introduce the elasticity $\mathscr{M}$-tensors and the nonsingular elasticity $\mathscr{M}$-tensors in Section 5, and prove their M-positive (semi)definiteness and propose other equivalent definitions for nonsingular elasticity $\mathscr{M}$-tensors.
Finally, conclusion remarks will be drawn in Section 6.

\section{Strong ellipticity and positive definiteness}\label{sec_se}

The tensor of elastic moduli for a linearly elastic material represented in a Cartesian coordinate system is a fourth-order three-dimensional tensor $\mathscr{A} = (a_{ijkl}) \in \mathbb{R}^{3 \times 3 \times 3 \times 3}$ which is invariant under the following permutations of indices
\begin{equation}\label{eq_sym}
  a_{ijkl} = a_{jikl} = a_{ijlk}.
\end{equation}
We use $\mathbb{E}_{4,n}$ to denote the set of all fourth-order $n$-dimensional tensors satisfying \eqref{eq_sym}, where $\mathbb{E}_{4,3}$ is exactly the set of all elasticity tensors.
The {\bf strong ellipticity condition (SE-condition)} for a tensor in $\mathbb{E}_{4,n}$ is stated by
\begin{equation}\label{eq_se}
  \mathscr{A} {\bf x}^2 {\bf y}^2 := \sum_{i,j,k,l=1}^n a_{ijkl} x_i x_j y_k y_l > 0
\end{equation}
for any nonzero vectors ${\bf x}, {\bf y} \in \mathbb{R}^n$.

The SE-condition equivalently requires that the optimal value of the following minimization problem is positive:
\begin{equation}\label{eq_min}
\begin{array}{cl}
  \min & \mathscr{A} {\bf x}^2 {\bf y}^2, \\
  {\rm s.t.} & {\bf x}^\top {\bf x} = 1,\ {\bf y}^\top {\bf y} = 1.
\end{array}
\end{equation}
The KKT condition \cite{Bazarra2013} of the minimization problem \eqref{eq_min} can be written as
\begin{equation}\label{eq_meig}
\left\{
\begin{array}{l}
  \mathscr{A} {\bf x} {\bf y}^2 = \lambda {\bf x}, \\
  \mathscr{A} {\bf x}^2 {\bf y} = \lambda {\bf y}, \\
  {\bf x}^\top {\bf x} = 1,\ {\bf y}^\top {\bf y} = 1,
\end{array}
\right.
\end{equation}
where $(\mathscr{A} {\bf x} {\bf y}^2)_i := \sum_{j,k,l=1}^n a_{ijkl} x_j y_k y_l$ and $(\mathscr{A} {\bf x}^2 {\bf y})_l := \sum_{i,j,k=1}^n a_{ijkl} x_i x_j y_k$.
In this formulation, Qi, Dai, and Han \cite{QiDaiHan09} defined the scalar $\lambda \in \mathbb{R}$ and two vectors ${\bf x}, {\bf y} \in \mathbb{R}^n$ as an {\bf M-eigenvalue} and a pair of corresponding {\bf M-eigenvectors} of $\mathscr{A}$, respectively.
Thus, we also call a tensor satisfying the SE-condition to be {\bf M-positive definite (M-PD)} \cite{QiLuo17}.
Similarly, a tensor $\mathscr{A} \in \mathbb{E}_{4,n}$ is said to be {\bf M-positive semidefinite (M-PSD)} \cite{QiLuo17} if $\mathscr{A} {\bf x}^2 {\bf y}^2 \geq 0$ for any vectors ${\bf x}, {\bf y} \in \mathbb{R}^n$.
The following theorem reveals that the M-positive definiteness is equivalent to the positivity of a tensor's M-eigenvalues.

\begin{theorem}[\cite{QiDaiHan09}]
  A tensor in $\mathbb{E}_{4,n}$ is M-positive definite if and only if all of its M-eigenvalues are positive;
  A tensor in $\mathbb{E}_{4,n}$ is M-positive semidefinite if and only if all of its M-eigenvalues are nonnegative.
\end{theorem}

We define a special tensor $\mathscr{E} \in \mathbb{E}_{4,n}$ by
$$
e_{ijkl} =
\left\{
\begin{array}{ll}
  1, & \text{if } i = j ~\text{and}~ k = l, \\
  0, & \text{otherwise},
\end{array}
\right.
$$
which serves as an  identity element in $\mathbb{E}_{4,n}$. We may call it the {\bf identity tensor} in this paper. When $n=3$, the components of the identity tensor $\mathscr{E}$ are $$e_{1111}= e_{1122}= e_{1133}= e_{2211}= e_{2222}= e_{2233}= e_{3311}= e_{3322}= e_{3333}=1,$$ and others are 0.
It can be verified that $\mathscr{E} {\bf x} {\bf y}^2 = {\bf x} ({\bf y}^\top {\bf y})$, $\mathscr{E} {\bf x}^2 {\bf y} = ({\bf x}^\top {\bf x}) {\bf y}$, and $\mathscr{E} {\bf x}^2 {\bf y}^2 = ({\bf x}^\top {\bf x}) ({\bf y}^\top {\bf y})$.
Hence, we have the following homogeneous definition for M-eigenvalues:
\begin{equation}\label{eq_meig2}
\left\{
\begin{array}{l}
  \mathscr{A} {\bf x} {\bf y}^2 = \lambda \mathscr{E} {\bf x} {\bf y}^2, \\
  \mathscr{A} {\bf x}^2 {\bf y} = \lambda \mathscr{E} {\bf x}^2 {\bf y}.
\end{array}
\right.
\end{equation}
Comparing \eqref{eq_meig} and \eqref{eq_meig2}, we can see that if the triplet $(\lambda,{\bf x},{\bf y})$ satisfies \eqref{eq_meig} then $(\lambda,\alpha{\bf x},\beta{\bf y})$ satisfies \eqref{eq_meig2} for any nonzero real scalar $\alpha,\beta$.
We note that \eqref{eq_meig2} is exactly the KKT condition of the following minimization problem:
$$
\begin{array}{cl}
  \min & \mathscr{A} {\bf x}^2 {\bf y}^2, \\
  {\rm s.t.} & ({\bf x}^\top {\bf x}) ({\bf y}^\top {\bf y}) = 1,
\end{array}
$$
whose optimal value being positive also guarantees the SE-condition.
The following proposition is an observation from the definition of the identity tensor.
\begin{proposition}\label{Pro2.2}
  Let $\mathscr{A}\in \mathbb{E}_{4,n}$.
  Suppose that $\mathscr{B} = \alpha(\mathscr{A}+\beta\mathscr{E})$, where $\alpha,\beta$ are two real scalars.
  Then $\mu$ is an M-eigenvalue of $\mathscr{B}$ if and only if $\mu = \alpha(\lambda + \beta)$ and $\lambda$ is an M-eigenvalue of $\mathscr{A}$.
  Furthermore, $\lambda$ and $\mu$ correspond to the same M-eigenvectors.
\end{proposition}
\begin{proof}
On the one hand, if  $\mu$ is an M-eigenvalue of $\mathscr{B}$, then $\mathscr{B} {\bf x} {\bf y}^2 = \mu {\bf x}$ and $ \mathscr{B} {\bf x} {\bf y}^2 = \mathscr{B} {\bf x}^2 {\bf y} = \mu {\bf y}$, where ${\bf x}, {\bf y}\in \mathbb{R}^n$ are the corresponding M-eigenvectors. When $\alpha=0$, the results is obvious. When $\alpha \neq 0$,
$$\alpha \mathscr{A}{\bf x}^2{\bf y} + \alpha\beta{\bf y} = \alpha\mathscr{A}{\bf x}^2{\bf y} + \alpha\beta\mathscr{E}{\bf x}^2{\bf y} = \alpha(\mathscr{A}+\beta\mathscr{E}){\bf x}^2{\bf y} = \mu {\bf y}$$
implies that
$$
\mathscr{A}{\bf x}^2{\bf y} = (\alpha^{-1}\mu-\beta){\bf y}.
$$
Similarly, we have
$$
\mathscr{A}{\bf x}{\bf y}^2 = (\alpha^{-1}\mu-\beta){\bf x}.
$$
Thus $\lambda=\alpha^{-1}\mu-\beta$ is the M-eigenvalue for $ \mathscr{A}$ corresponding to ${\bf x}$ and ${\bf y}$.

On the other hand, when $\mu = \alpha(\lambda + \beta)$ and $\lambda$ is an M-eigenvalue of $\mathscr{A}$ corresponding to the M-eigenvectors ${\bf x}$ and ${\bf y}$, it can get
$$\mu {\bf y}= \alpha(\lambda + \beta){\bf y}=  \alpha \lambda {\bf y}  + \alpha \beta{\bf y}=  \alpha \mathscr{A}{\bf x}^2{\bf y} + \alpha \beta\mathscr{E}{\bf x}^2{\bf y}= \alpha(\mathscr{A} + \beta\mathscr{E}) {\bf x}^2{\bf y} = \mathscr{B}{\bf x}^2{\bf y},$$
$$\mu {\bf x}= \alpha(\lambda + \beta){\bf x}=  \alpha \lambda {\bf x}  + \alpha \beta{\bf x}=  \alpha \mathscr{A}{\bf x}{\bf y}^2 + \alpha \beta\mathscr{E}{\bf x}{\bf y}^2= \alpha(\mathscr{A} + \beta\mathscr{E}) {\bf x}{\bf y}^2 = \mathscr{B}{\bf x}{\bf y}^2.$$
Hence, $\mu$ is an M-eigenvalue of $\mathscr{B}$ corresponding to the same M-eigenvectors ${\bf x}$ and ${\bf y}$.
\end{proof}

There are two common ways to unfold a tensor in $\mathbb{E}_{4,n}$ into  $n^2$-by-$n^2$ matrices:
\begin{enumerate}[(i)]
  \item ${\bf A}_x =
       \begin{bmatrix}
       {\bf A}_x^{(1,1)} & {\bf A}_x^{(1,2)} & \cdots & {\bf A}_x^{(1,n)} \\
       {\bf A}_x^{(2,1)} & {\bf A}_x^{(2,2)} & \cdots & {\bf A}_x^{(2,n)} \\
       \vdots & \vdots & \ddots & \vdots \\
       {\bf A}_x^{(n,1)} & {\bf A}_x^{(n,2)} & \cdots & {\bf A}_x^{(n,n)}
       \end{bmatrix} \in \mathbb{R}^{n^2 \times n^2}$,
  \item ${\bf A}_y =
       \begin{bmatrix}
       {\bf A}_y^{(1,1)} & {\bf A}_y^{(1,2)} & \cdots & {\bf A}_y^{(1,n)} \\
       {\bf A}_y^{(2,1)} & {\bf A}_y^{(2,2)} & \cdots & {\bf A}_y^{(2,n)} \\
       \vdots & \vdots & \ddots & \vdots \\
       {\bf A}_y^{(n,1)} & {\bf A}_y^{(n,2)} & \cdots & {\bf A}_y^{(n,n)}
       \end{bmatrix} \in \mathbb{R}^{n^2 \times n^2}$,
\end{enumerate}
where ${\bf A}_x^{(k,l)}:= \mathscr{A}(:,:,k,l)$($k,l = 1,\cdots,n$) and ${\bf A}_y^{(i,j)}:= \mathscr{A}(i,j,:,:)$($i,j = 1,\cdots,n$).
Note that ${\bf A}_x$ and ${\bf A}_y$ are permutation similar to each other, i.e., there is a permutation matrix ${\bf P}$ such that ${\bf A}_x = {\bf P}^\top {\bf A}_y {\bf P}$.
Then $\mathscr{A}$ is M-PD or M-PSD if ${\bf A}_x$ (or equivalently ${\bf A}_y$) is PD or PSD, respectively.
This can be proved by noticing that
$$
\mathscr{A} {\bf x}^2 {\bf y}^2
= ({\bf y} \otimes {\bf x})^\top {\bf A}_x ({\bf y} \otimes {\bf x})
= ({\bf x} \otimes {\bf y})^\top {\bf A}_y ({\bf x} \otimes {\bf y}),
$$
where $\otimes$ denotes the Kronecker product \cite{HornJohnson13}.
Thus we call $\mathscr{A}$ S-positive (semi)definite if ${\bf A}_x$ or ${\bf A}_y$ is positive (semi)definite, and call the eigenvalues of ${\bf A}_x$ or ${\bf A}_y$ the S-eigenvalues of $\mathscr{A}$.
The S-positive definiteness is a sufficient condition for the M-positive definiteness, but the converse is not true. A counter example is as follows.
\begin{example}
  Consider the case $n=3$. Let $\mathscr{A} \in \mathbb{E}_{4,3}$ be defined by
  $$
    a_{1111}= a_{2222}= a_{3333}=2, ~~a_{1221}= a_{2121}= a_{2112}= a_{1212}=1,
  $$
  and all other entries equal to zero. Then we have
  $$\mathscr{A} {\bf x}^2 {\bf y}^2=2(x_1y_1+x_2y_2)^2+2x_3^2y_3^2,$$
  thus $\mathscr{A}$ is M-PSD apparently, while the unfolding matrix
  $$
  {\bf A}_x = {\bf A}_y =
  \left[
  \begin{array}{ccc|ccc|ccc}
   2 & 0 & 0 & 0 & 1 & 0 & 0 & 0 & 0 \\
   0 & 0 & 0 & 1 & 0 & 0 & 0 & 0 & 0 \\
   0 & 0 & 0 & 0 & 0 & 0 & 0 & 0 & 0 \\ \hline
   0 & 1 & 0 & 0 & 0 & 0 & 0 & 0 & 0 \\
   1 & 0 & 0 & 0 & 2 & 0 & 0 & 0 & 0 \\
   0 & 0 & 0 & 0 & 0 & 0 & 0 & 0 & 0 \\ \hline
   0 & 0 & 0 & 0 & 0 & 0 & 0 & 0 & 0 \\
   0 & 0 & 0 & 0 & 0 & 0 & 0 & 0 & 0 \\
   0 & 0 & 0 & 0 & 0 & 0 & 0 & 0 & 2 \\
  \end{array}
  \right]
  $$
 is not positive semidefinite.
\end{example}

Next, we need to introduce several more notations for convenience.
Let $\mathscr{A} \in \mathbb{E}_{4,n}$ and ${\bf x},{\bf y} \in \mathbb{R}^n$.
We define two $n$-by-$n$ matrices $\mathscr{A} {\bf x}^2  \in \mathbb{R}^{n \times n}$ and $\mathscr{A} {\bf y}^2 \in \mathbb{R}^{n \times n}$ by
\[
\begin{split}
  (\mathscr{A} {\bf x}^2 )_{kl} &:= \sum_{i,j=1}^n a_{ijkl} x_i x_j, \quad k,l = 1,2,\dots,n, \\
  (\mathscr{A} {\bf y}^2)_{ij} &:= \sum_{k,l=1}^n a_{ijkl} y_k y_l, \quad i,j = 1,2,\dots,n.
\end{split}
\]
We note that
$$
\mathscr{A} {\bf x}^2
=\begin{bmatrix}
  {\bf x}^\top {\bf A}_x^{(1,1)} {\bf x} & {\bf x}^\top {\bf A}_x^{(1,2)} {\bf x} & \cdots & {\bf x}^\top {\bf A}_x^{(1,n)} {\bf x} \\
  {\bf x}^\top {\bf A}_x^{(2,1)} {\bf x} & {\bf x}^\top {\bf A}_x^{(2,2)} {\bf x} & \cdots & {\bf x}^\top {\bf A}_x^{(2,n)} {\bf x} \\
  \vdots & \vdots & \ddots & \vdots \\
  {\bf x}^\top {\bf A}_x^{(n,1)} {\bf x} & {\bf x}^\top {\bf A}_x^{(n,2)} {\bf x} & \cdots & {\bf x}^\top {\bf A}_x^{(n,n)} {\bf x}
\end{bmatrix},
$$ and
$$
\mathscr{A} {\bf y}^2
=\begin{bmatrix}
  {\bf y}^\top {\bf A}_y^{(1,1)} {\bf y} & {\bf y}^\top {\bf A}_y^{(1,2)} {\bf y} & \cdots & {\bf y}^\top {\bf A}_y^{(1,n)} {\bf y} \\
  {\bf y}^\top {\bf A}_y^{(2,1)} {\bf y} & {\bf y}^\top {\bf A}_y^{(2,2)} {\bf y} & \cdots & {\bf y}^\top {\bf A}_y^{(2,n)} {\bf y} \\
  \vdots & \vdots & \ddots & \vdots \\
  {\bf y}^\top {\bf A}_y^{(n,1)} {\bf y} & {\bf y}^\top {\bf A}_y^{(n,2)} {\bf y} & \cdots & {\bf y}^\top {\bf A}_y^{(n,n)} {\bf y}
\end{bmatrix}.
$$

Furthermore, it is straightforward to verify that
\begin{equation}\label{eq_Axx}
\begin{split}
  &\mathscr{A} {\bf x}^2 {\bf y}^2 = {\bf y}^\top (\mathscr{A} {\bf x}^2 ) {\bf y} = {\bf x}^\top (\mathscr{A}  {\bf y}^2) {\bf x}, \\
  &\mathscr{A} {\bf x}^2 {\bf y} = (\mathscr{A} {\bf x}^2) {\bf y}, \quad \mathscr{A} {\bf x} {\bf y}^2 = (\mathscr{A}  {\bf y}^2) {\bf x}.
\end{split}
\end{equation}

The symmetries in $\mathscr{A}$ imply that both  $\mathscr{A} {\bf x}^2 $ and $\mathscr{A} {\bf y}^2$ are symmetric matrix. According to Eq.\eqref{eq_Axx}, we can prove the following necessary and sufficient condition for the M-positive (semi)definiteness.
\begin{proposition}
  Let $\mathscr{A} \in \mathbb{E}_{4,n}$. Then $\mathscr{A}$ is M-PD or M-PSD if and only if the matrix $\mathscr{A} {\bf x}^2 $ (or $\mathscr{A} {\bf y}^2$) is PD or PSD for each nonzero ${\bf x}\in \mathbb{R}^n$ (or ${\bf y}\in \mathbb{R}^n$), respectively.
\end{proposition}
\begin{proof}
 On one side, if $\mathscr{A}$ is M-PD, then for any ${\bf x}, {\bf y} \in \mathbb{R}^n\backslash\{{\bf 0}\}$, $\mathscr{A} {\bf x}^2 {\bf y}^2 > 0$. This means that ${\bf y}^{\top}(\mathscr{A} {\bf x}^2){\bf y} > 0$, for  any nonzero ${\bf y} \in \mathbb{R}^n$. Hence, the matrix $\mathscr{A} {\bf x}^2 $ is PD for each nonzero ${\bf x}\in \mathbb{R}^n$. Similarity, $\mathscr{A} {\bf y}^2 $ is PD for each nonzero ${\bf y}\in \mathbb{R}^n$. On the other side, when $\mathscr{A} {\bf x}^2 $ (or $\mathscr{A} {\bf y}^2$) is PD for each nonzero ${\bf x}\in \mathbb{R}^n$ (or ${\bf y}\in \mathbb{R}^n$), it has ${\bf y}^{\top}(\mathscr{A} {\bf x}^2){\bf y} > 0$ (or ${\bf x}^{\top}(\mathscr{A} {\bf y}^2){\bf x} > 0$ ) for any nonzero ${\bf y} \in \mathbb{R}^n$ (or ${\bf x}\in \mathbb{R}^n$), which means that $\mathscr{A}$ is M-PD. \\
 Similarity, $\mathscr{A}$ is M-PSD if and only if the matrix $\mathscr{A} {\bf x}^2 $ (or $\mathscr{A} {\bf y}^2$) is PSD for each ${\bf x}\in \mathbb{R}^n$ (or ${\bf y}\in \mathbb{R}^n$), respectively.
\end{proof}

Generally speaking, the above necessary and sufficient condition is as hard as the SE-condition to check. However, it motivates some checkable sufficient conditions. Hence, we will introduce another sufficient condition to verify the SE-condition.

\section{A sufficient condition with a verification algorithm}\label{sec_suf}

Recall that every positive semidefinite matrix can be decomposed into the sum of some rank-one positive semidefinite matrices and the minimal number of terms is exactly its rank \cite{HornJohnson13}.
Thus we have the following sufficient condition for a tensor $\mathscr{A}$ to be M-PSD that
\begin{equation}\label{eq_rank1+}
  \mathscr{A} {\bf y}^2 = \sum_{s = 1}^r \alpha_s {\bf f}_s({\bf y}) {\bf f}_s({\bf y})^\top, \quad \alpha_s > 0,
\end{equation}
where each ${\bf f}_s({\bf y})$ is a homogeneous function of degree one, i.e., ${\bf f}_s({\bf y}) = {\bf U}_s {\bf y}$ for $s = 1,2,\dots,r$.
Obviously, any matrix $\mathscr{A} {\bf y}^2$ in the above form is PSD and thus $\mathscr{A}$ is M-PSD.
Furthermore, if $\mathscr{A}$ is M-PD, then the number of terms in the summation should be no less than $n$, i.e., $r \geq n$.
Denote the entries of each ${\bf U}_s$ as $u_{ij}^{(s)}$ ($i,j=1,\cdots,n$).
Then \eqref{eq_rank1+} reads
\[
\begin{split}
\sum_{k,l=1}^n a_{ijkl} y_k y_l
&= \sum_{s=1}^r \alpha_s \Big( \sum_{k=1}^n u_{ik}^{(s)} y_k \Big) \Big( \sum_{l=1}^n u_{jl}^{(s)} y_l \Big) \\
&= \sum_{k,l=1}^n \Big( \sum_{s=1}^r \alpha_s u_{ik}^{(s)} u_{jl}^{(s)} \Big) y_k y_l.
\end{split}
\]
Therefore, given ${\bf U}_s$ ($s = 1,2,\dots,r$), the entries of $\mathscr{A}$ are uniquely determined by
\begin{equation}\label{eq_con}
  a_{ijkl} = \frac{1}{2} \sum_{s=1}^r \alpha_s \Big( u_{ik}^{(s)} u_{jl}^{(s)} + u_{jk}^{(s)} u_{il}^{(s)} \Big),
\end{equation}
which satisfies the symmetries $a_{ijkl} = a_{jikl} = a_{ijlk}$.

Next, we shall discuss when a tensor in $\mathbb{E}_{4,n}$ can be represented by \eqref{eq_con}.
Denote another fourth-order three-dimensional tensor $\mathscr{B}$ by
$$
b_{ijkl} = \sum_{s=1}^r \alpha_s u_{ik}^{(s)} u_{jl}^{(s)}.
$$
Note that $\mathscr{B}$ may not in the set $\mathbb{E}_{4,n}$, i.e., it is not required to obey \eqref{eq_sym}, but it still satisfies a weaker symmetry that $b_{ijkl} = b_{jilk}$.
It can be seen that its unfolding ${\bf B}$ is a PSD matrix from
$$
{\bf B} = \sum_{s=1}^r \alpha_s {\bf u}_s {\bf u}_s^\top,
$$
where ${\bf u}_s$ is the unfolding (or called vectorization) of ${\bf U}_s$ ($s = 1,2,\dots,r$).
Hence $\mathscr{B}$ is S-PSD since all the coefficients $\alpha_s$ are positive.
Furthermore, comparing the entries of $\mathscr{A}$ and $\mathscr{B}$, we will find that
$$
a_{ijkl} = a_{jikl} = \frac{1}{2} (b_{ijkl} + b_{jikl}), \quad i,j,k,l = 1,2,3,
$$
and thus
$$
\mathscr{A} {\bf x}^2 {\bf y}^2 = \mathscr{B} {\bf x}^2 {\bf y}^2 = \mathscr{B} \big( {\bf x} {\bf y}^\top \big)^2.
$$
Therefore $\mathscr{A}$ is M-PD or M-PSD when $\mathscr{B}$ is S-PD or S-PSD, respectively.

Given a fourth-order tensor $\mathscr{A} \in \mathbb{E}_{4,n}$, we denote
$$
\mathbb{T}_{\mathscr{A}} := \{ \mathscr{T}:\, t_{ijkl} = t_{jilk},\, t_{ijkl} + t_{jikl} = 2 a_{ijkl} \}.
$$
We also denote the set of all fourth-order S-PSD tensors as
$$
\mathbb{S} := \{ \mathscr{T}:\, t_{ijkl} = t_{jilk},\, \mathscr{T} \text{ is S-PSD} \}.
$$
Note that both $\mathbb{T}_{\mathscr{A}}$ and $\mathbb{S}$ are closed convex sets, where $\mathbb{T}_{\mathscr{A}}$ is a linear subspace of the whole space of all the fourth-order three-dimensional tensor with $t_{ijkl} = t_{jilk}$ and $\mathbb{S}$ is isomorphic with the nine-by-nine symmetric PSD matrix cone.
Furthermore, we have actually proved the following sufficient condition for a tensor to be M-PD or M-PSD.
\begin{theorem}
  Let $\mathscr{A} \in \mathbb{E}_{4,n}$.
  If $\mathbb{T}_{\mathscr{A}} \cap \mathbb{S} \neq \emptyset$, then $\mathscr{A}$ is M-PSD;
  If $\mathbb{T}_{\mathscr{A}} \cap (\mathbb{S} \setminus \partial\mathbb{S}) \neq \emptyset$, then $\mathscr{A}$ is M-PD.
\end{theorem}

A method called projections onto convex sets (POCS) \cite{BauschkeBorwein96,EscalanteRaydan11} is often employed to check whether the intersection of two closed convex sets is empty or not.
POCS is also known as the alternating projection algorithm.
Denote ${\cal P}_1$ and ${\cal P}_2$ as the projections onto $\mathbb{T}_{\mathscr{A}}$ and $\mathbb{S}$, respectively.
Then POCS is stated by
$$
\left\{
\begin{array}{l}
  \mathscr{B}^{(t+1)} = {\cal P}_2( \mathscr{A}^{(t)} ), \\
  \mathscr{A}^{(t+1)} = {\cal P}_1( \mathscr{B}^{(t+1)} ),
\end{array}
\right.
\quad t = 0,1,2,\dots.
$$
The algorithm can be described as the following iterative scheme:
\begin{equation}\label{alg_pocs}
\left\{
\begin{array}{l}
  \text{Eigendecomposition } {\bf A}^{(t)} = {\bf V}^{(t)} {\bf D}^{(t)} ({\bf V}^{(t)})^\top, \\
  {\bf B}^{(t+1)} = {\bf V}^{(t)} {\bf D}_+^{(t)} ({\bf V}^{(t)})^\top, \\
  a_{iikl}^{(t+1)} = a_{iikl} \text{ for } i,k,l = 1,\cdots,n, \\
  a_{ijkk}^{(t+1)} = a_{ijkk} \text{ for } i,j,k = 1,\cdots,n, \\
  a_{ijkl}^{(t+1)} = a_{ijkl} + \frac{1}{2} (b_{ijkl}^{(t+1)} - b_{jikl}^{(t+1)}) \text{ for } i \neq j,\ k \neq l,
\end{array}
\right.
\quad t = 0,1,2,\dots.
\end{equation}
where $\mathscr{A}^{(0)} = \mathscr{A}$, ${\bf A}^{(t)}$ and ${\bf B}^{(t)}$ are the unfolding matrices of $\mathscr{A}^{(t)}$ and $\mathscr{B}^{(t)}$ respectively, and ${\bf D}_+^{(t)} = {\rm diag}\big( \max(d_{ii}^{(t)}, 0) \big)$.
The convergence of the alternating projection method between two closed convex sets is known for a long time \cite{CheneyGoldstein59}.

\begin{theorem}
  Let $\mathscr{A} \in \mathbb{E}_{4,n}$. If $\mathbb{T}_{\mathscr{A}} \cap \mathbb{S} \neq \emptyset$, then the sequences $\{ \mathscr{A}^{(t)} \}$ and $\{ \mathscr{B}^{(t)} \}$ produced by Algorithm \eqref{alg_pocs} both converge to a point $\mathscr{A}^\ast \in \mathbb{T}_{\mathscr{A}} \cap \mathbb{S}$.
\end{theorem}

Because the convergence of POCS requires the involved convex sets to be closed, Algorithm \eqref{alg_pocs} is only suitable for identifying the M-positive semidefiniteness.
If we want to check the M-positive definiteness, then some modifications are needed.
Note that $\mathscr{E} {\bf x}^2 {\bf y}^2 = ({\bf x}^\top {\bf x}) ({\bf y}^\top {\bf y})$.
Hence $\mathscr{E}$ is M-PD, which implies that $\mathscr{A}$ is M-PD if and only if $\mathscr{A} - \epsilon \mathscr{E}$ is M-PSD for some sufficiently small $\epsilon > 0$.
From such observation, we can apply POCS to $\mathscr{A} - \epsilon \mathscr{E}$ with a very small $\epsilon$.
If the iteration converges and both $\{ \mathscr{A}^{(t)} \}$ and $\{ \mathscr{B}^{(t)} \}$ converge to the same tensor, then we can conclude that $\mathscr{A}$ is M-PD, i.e., the strong ellipticity holds.

\section{Nonnegative elasticity tensors}\label{sec_nonneg}

We have a well-known theory about nonnegative matrices called the Perron-Frobenius theorem \cite{BermanPlemmons94}, which states that the spectral radius of any nonnegative matrix is an eigenvalue with a nonnegative eigenvector and the eigenvector is positive and unique if the matrix is irreducible.
In the past decades, the Perron-Frobenius theorem has been extended to higher order tensors by Chang, Pearson, and Zhang \cite{ChangPearsonZhang08} and Yang and Yang \cite{YangYang11,YangYang10}.
One may refer to \cite[Chapter 3]{QiLuo17} for a whole picture of the nonnegative tensor theory.
We will also obtain similar results for nonnegative elasticity tensors in this section.

From the discussions in Section \ref{sec_se}, we have variational forms of the extremal M-eigenvalues.
Let $\mathscr{B} \in \mathbb{E}_{4,n}$.
Denote $\lambda_{\max}(\mathscr{B})$ and $\lambda_{\min}(\mathscr{B})$ as the maximal and the minimal M-eigenvalues of $\mathscr{B}$, respectively.
Then
\begin{equation}\label{eq_maxmin}
\begin{split}
  \lambda_{\max}(\mathscr{B}) &= \max \big\{ \mathscr{B} {\bf x}^2 {\bf y}^2:\, {\bf x},{\bf y} \in \mathbb{R}^n,\, {\bf x}^{\top}{\bf x}={\bf y}^{\top}{\bf y}=1 \big\}, \\
  \lambda_{\min}(\mathscr{B}) &= \min \big\{ \mathscr{B} {\bf x}^2 {\bf y}^2:\, {\bf x},{\bf y} \in \mathbb{R}^n,\, {\bf x}^{\top}{\bf x}={\bf y}^{\top}{\bf y}=1 \big\}.
\end{split}
\end{equation}
The maximal absolute value of all the M-eigenvalues is called the M-spectral radius of a tensor in $\mathbb{E}_{4,n}$, denoted by $\rho_M(\cdot)$.
Apparently, the M-spectral radius is equal to the greater one of the absolute values of the maximal and the minimal M-eigenvalues.
The following theorem reveals that $\rho_M(\mathscr{B}) = \lambda_{\max}(\mathscr{B})$ when $\mathscr{B} \in \mathbb{E}_{4,n}$ is nonnegative.

\begin{theorem}\label{thm_nn}
  The M-spectral radius of any nonnegative tensor in $\mathbb{E}_{4,n}$ is exactly its greatest M-eigenvalue. Furthermore, there is a pair of nonnegative M-eigenvectors corresponding to the M-spectral radius.
\end{theorem}
\begin{proof}
  It enough to show that $\lambda_{\max}(\mathscr{B}) \geq \big| \lambda_{\min}(\mathscr{B}) \big|$ for proving the first statement.
  For convenience, denote $\lambda_1$ and $\lambda_2$ as the maximal and the minimal M-eigenvalues of $\mathscr{B}$ respectively, and $({\bf x}_1,{\bf y}_1)$ and $({\bf x}_2,{\bf y}_2)$ are the corresponding M-eigenvectors.
  By \eqref{eq_maxmin}, we know that $\lambda_1 = \mathscr{B} {\bf x}_1^2 {\bf y}_1^2$ and $\lambda_2 = \mathscr{B} {\bf x}_2^2 {\bf y}_2^2$.
  Then employing the nonnegativity of the entries of $\mathscr{B}$, we have
  $$
  |\lambda_2| = \big| \mathscr{B} {\bf x}_2^2 {\bf y}_2^2 \big|
  \leq \mathscr{B} |{\bf x}_2|^2 |{\bf y}_2|^2 \leq \lambda_1.
  $$

  Next, we consider the eigenvectors of the M-spectral radius.
  Assume that ${\bf x}_1$ or ${\bf y}_1$ is not a nonnegative vector.
  Then we also have
  $$
  \lambda_1 = \mathscr{B} {\bf x}_1^2 {\bf y}_1^2
  \leq \mathscr{B} |{\bf x}_1|^2 |{\bf y}_1|^2 \leq \lambda_1,
  $$
  thus $\mathscr{B} |{\bf x}_1|^2 |{\bf y}_1|^2 = \lambda_1$.
  Therefore, $\big( |{\bf x}_1|,|{\bf y}_1| \big)$ is also a pair of M-eigenvectors corresponding to $\lambda_1$, which is nonnegative.
\end{proof}

Theorem \ref{thm_nn} can be regarded as the weak Perron-Frobenius theorem for the tensors in $\mathbb{E}_{4,n}$.
Combining Theorem \ref{thm_nn} and \eqref{eq_maxmin}, we have the following corollary, which shrinks the feasible domain in \eqref{eq_maxmin}.

\begin{corollary}\label{coro nn}
  Let $\mathscr{B} \in \mathbb{E}_{4,n}$. If $\mathscr{B}$ is nonnegative, then
  $$
  \rho(\mathscr{B}) = \max \big\{ \mathscr{B} {\bf x}^2 {\bf y}^2:\, {\bf x},{\bf y} \in \mathbb{R}_+^n,\, {\bf x}^{\top}{\bf x}={\bf y}^{\top}{\bf y}=1 \big\}.
  $$
\end{corollary}

\begin{corollary}\label{coro 3.3}
  Let $\mathscr{B} \in \mathbb{E}_{4,n}$  be nonnegative. Then we have  $\rho(\mathscr{B}) = 0$ if and only if $\mathscr{B}=\mathscr{O}$, where $\mathscr{O}$ is a zero tensor.
\end{corollary}
\begin{proof}
On the one hand, if $\mathscr{B}$ is a zero tensor, then $\rho(\mathscr{B}) = 0$.

On the other hand, if $\rho(\mathscr{B}) = 0$, then we have $ \mathscr{B} {\bf x}^2 {\bf y}^2 = 0,$ for any ${\bf x},{\bf y}\in \mathbb{R}^n$. This means that, for any ${\bf y}\in \mathbb{R}^n$, we can get $ {\bf y}^{\top} (\mathscr{B} {\bf x}^2 ){\bf y}= 0.$ Hence, for any ${\bf x}\in \mathbb{R}^n$ and fixed $k,l \in \{1,2,\cdots,n\}$, $\mathscr{B} {\bf x}^2 = \sum_{i,j=1}^{n} b_{ijkl}x_{i}x_{j} = {\bf O}$, where ${\bf O}$ is a zero matrix.

 Assume that ${\bf x} = {\bf e}_i$ ($i=1,\cdots,n$), where ${\bf e}_i$ is the unit vectors whose $i$-th component is 1 and others are zero, we have $b_{iikl}=0$ for all fixed $k,l \in \{1,2,\cdots,n\}$.
 Furthermore, when ${\bf x} = {\bf e}_i +  {\bf e}_j$ ($i\neq j$, and $i,j=1,\cdots,n$), we have $b_{iikl}+ 2b_{ijkl}+b_{jjkl}=0$ for all fixed $k,l \in \{1,2,\cdots,n\}$. Thus we have  $b_{ijkl}=0$ for all fixed $k,l \in \{1,2,\cdots,n\}$. This means that $\mathscr{B}=\mathscr{O}$ when $\rho(\mathscr{B}) = 0$.

\end{proof}

Chang, Qi, and Zhou \cite{ChangQiZhou10} also studied the strong ellipticity for nonnegative elasticity tensors. They introduced the singular values of a tensor $\mathscr{B} \in \mathbb{E}_{4,n}$ as
$$
\left\{
\begin{array}{l}
  \mathscr{B} {\bf x} {\bf y}^2 = \sigma {\bf x}^{[3]}, \\
  \mathscr{B} {\bf x}^2 {\bf y} = \sigma {\bf y}^{[3]},
\end{array}
\right.
$$
and they also investigate the Perron-Frobenius theorem for the singular values.
Nevertheless, it is hard to find an identity tensor similar to the tensor $\mathscr{E}$ in our case, thus we may not be able to define a kind of $\mathscr{M}$-tensors with respect to their singular values. However, they introduced the definition for irreducibility of the elasticity tensors in \cite{ChangQiZhou10}. Recall the notations of  ${\bf A}_x^{(k,l)}$ ($k,l = 1,\cdots,n$) and ${\bf A}_y^{(i,j)}$ ($i,j = 1,\cdots,n$) in Section \ref{sec_se}.  Let $\mathscr{B} \in \mathbb{E}_{4,n}$ be nonnegative. If all the $n \times n$ matrices ${\bf B}_x^{(k,k)}$ ($k = 1,\cdots,n$) and ${\bf B}_y^{(i,i)}$ ($i = 1,\cdots,n$) are all irreducible matrices, then the nonnegative elasticity tensor $\mathscr{B}$ is called {\bf irreducible}\cite{ChangQiZhou10}. With the irreducibility of an elasticity nonnegative tensor, a useful lemma can be proved.

\begin{lemma}\label{lem1216}
Let $\mathscr{B} \in \mathbb{E}_{4,n}$. If $\mathscr{B}$ is nonnegative and irreducible, then there is a pair of positive M-eigenvectors corresponding to its M-spectral radius, i.e.,
  $$
  \rho(\mathscr{B}) = \max \big\{ \mathscr{B} {\bf x}^2 {\bf y}^2:\, {\bf x},{\bf y} \in \mathbb{R}_{++}^n,\, {\bf x}^{\top}{\bf x}={\bf y}^{\top}{\bf y}=1 \big\},
  $$
  where $\mathbb{R}_{++}^n$ is positive real vector field with dimension $n$.
\end{lemma}
\begin{proof}
  Since  $\mathscr{B}$ is nonnegative, from Theorem \ref{thm_nn}, there exists a pair of nonnegative M-eigenvectors ${\bf x}, {\bf y}\in \mathbb{R}^n_{+}$ corresponding to its M-spectral radius $\rho(\mathscr{B})$. Moreover, we have
  $$\mathscr{B} {\bf x}^2 = \sum_{i,j=1}^{n} {\bf B}_y^{(i,j)}x_{i}x_{j} \geq  \sum_{i=1}^{n} {\bf B}_y^{(i,i)}x_{i}x_{i}.$$
  Due to the nonnegativity of ${\bf x}$, there exists an $i_{0}\in \{1,\cdots,n\}$ such that $x_{i_{0}} > 0$. Hence,
  $$\mathscr{B} {\bf x}^2  \geq  \sum_{i=1}^{n} {\bf B}_y^{(i,i)}x_{i}x_{i} \geq {\bf B}_y^{(i_{0},i_{0})}x_{i_{0}}x_{i_{0}}.$$
  Since $\mathscr{B}$ is irreducible,  ${\bf B}_y^{(i,i)}$ ($i \in \{1,\cdots,n\}$) are irreducible matrices. Thus, $\mathscr{B} {\bf x}^2$ is also irreducible.
  Hence, the corresponding M-eigenvector ${\bf y}$ is positive such that $$\mathscr{B} {\bf x}^2{\bf y} = \rho(\mathscr{B}){\bf y}.$$ Similarity,  $\mathscr{B} {\bf y}^2$ is irreducible, and the corresponding M-eigenvector ${\bf x}$ is positive such that $\mathscr{B} {\bf x}{\bf y}^2 = \rho(\mathscr{B}){\bf x}.$ In the summery, we have $$
  \rho(\mathscr{B}) = \max \big\{ \mathscr{B} {\bf x}^2 {\bf y}^2:\, {\bf x},{\bf y} \in \mathbb{R}_{++}^n,\, {\bf x}^{\top}{\bf x}={\bf y}^{\top}{\bf y}=1 \big\}.
  $$
\end{proof}

However, when $\mathscr{B}$ is nonnegative and irreducible,  not all positive M-eigenvectors are  corresponding to its M-spectral radius, such as following example:
\begin{example}
 Let $\mathscr{B} \in \mathbb{E}_{4,2}$ be defined by $b_{1111}= 4$, $b_{1122} =  b_{2211}= 10$, $b_{2222}= 2,$ $b_{1112} = b_{1121}= b_{1211} =b_{2111} =1,$
 $b_{1212} = b_{1221}= b_{2112} =b_{2121} =1,$ and $b_{1222} = b_{2122}= b_{2212} =b_{2221} =2.$

It is a nonnegative irreducible elasticity tensor. By computing its M-eigenvalues and corresponding M-eigenvectors, we have
$\lambda_{max}=10.9075$, and the M-eigenvectors are $${\bf x}_{1}=(0.2936, 0.9560)^{\top}, {\bf y}_{1}=(0.9442, 0.3294)^{\top},$$ and $${\bf x}_{2}=(0.9442, 0.3294)^{\top}, {\bf y}_{2}=(0.2936, 0.9560)^{\top}.$$ Furthermore, the second max M-eigenvalue is $\lambda_{2nd-max}=10.5$. The corresponding M-eigenvectors are ${\bf x}=(0.7071, 0.7071)^{\top}$, ${\bf y}=(0.7071, 0.7071)^{\top}$. Hence, not all the positive M-eigenvectors are corresponding to its M-spectral radius.
\end{example}

\section{Elasticity $\mathscr{M}$-tensors}\label{sec_elasM}

Recall that the identity tensor $\mathscr{E}$ is defined by $e_{iikk} = 1$ and other entries being zero.
Let $\mathscr{A} \in \mathbb{E}_{4,n}$.
Accordingly, we call the entries $a_{iikk}$ ($i,k = 1,2,\dots,n$) diagonal, and other entries are called off-diagonal.
Obviously, the diagonal entries of an M-PD tensor must be positive, and the ones of an M-PSD tensor must be nonnegative.
It is worth noting that the diagonal entries of an elasticity tensor also lie on the diagonal of its unfolding matrix.

A tensor in $\mathbb{E}_{4,n}$ is called an elasticity $\mathscr{Z}$-tensor if all its off-diagonal entries are non-positive.
If $\mathscr{A} \in \mathbb{E}_{4,n}$ is an elasticity $\mathscr{Z}$-tensor, then we can always write it as $\mathscr{A} = s \mathscr{E} - \mathscr{B}$, where $\mathscr{B}$ is a nonnegative tensor in $\mathbb{E}_{4,n}$.
Such partition of an elasticity $\mathscr{E}$-tensor is not unique.
If a tensor $\mathscr{A} \in \mathbb{E}_{4,n}$ can be written as $\mathscr{A} = s \mathscr{E} - \mathscr{B}$ satisfying that $\mathscr{B} \in \mathbb{E}_{4,n}$ is nonnegative and $s \geq \rho_M(\mathscr{B})$, then we call $\mathscr{A}$ an {\bf elasticity $\mathscr{M}$-tensor}. Furthermore, if $s > \rho_M(\mathscr{B})$, then we call $\mathscr{A}$ a {\bf nonsingular elasticity $\mathscr{M}$-tensor}.

\begin{theorem}\label{thm_mindiag}
  Let $\mathscr{A} \in \mathbb{E}_{4,n}$ be an elasticity $\mathscr{Z}$-tensor. Then $\mathscr{A}$ is a nonsingular elasticity $\mathscr{M}$-tensor if and only if $\alpha > \rho_M (\alpha \mathscr{E} - \mathscr{A})$, where $\alpha = \max \big\{ a_{iikk}:\, i,k=1,2,\dots,n \big\}$.
\end{theorem}
\begin{proof}
  The ``if'' part is obvious by the partition $\mathscr{A} = \alpha\mathscr{E} - (\alpha \mathscr{E} - \mathscr{A})$.
  Thus we focus on the ``only if'' part.
  If $\mathscr{A}$ is a nonsingular elasticity $\mathscr{M}$-tensor, then it can be written as $\mathscr{A} = s \mathscr{E} - \mathscr{B}$ satisfying that $\mathscr{B} \in \mathbb{E}_{4,n}$ is nonnegative and $s > \rho_M(\mathscr{B})$.
  Denote $\beta = \min \big\{ b_{iikk}:\, i,k=1,2,\dots,n \big\}$, then $\alpha = s-\beta$.
  Moreover, we can also write $\alpha \mathscr{E} - \mathscr{A} = \mathscr{B} - \beta \mathscr{E}$, thus $\rho_M(\alpha \mathscr{E} - \mathscr{A}) = \rho_M(\mathscr{B}) - \beta$.
  Therefore, $s > \rho_M(\mathscr{B})$ implies that $\alpha > \rho_M (\alpha \mathscr{E} - \mathscr{A})$.
\end{proof}

The above theorem is a simple but useful observation.
We can utilize this theorem to prove the following proposition, which reveals that any elasticity $\mathscr{M}$-tensor is the limit of a series of nonsingular elasticity $\mathscr{M}$-tensors.
Hence, we may omit the proofs of following results for general elasticity $\mathscr{M}$-tensors, since it can be verified by taking limits of the results for nonsingular elasticity $\mathscr{M}$-tensors.

\begin{proposition}
 $\mathscr{A} \in \mathbb{E}_{4,n}$ is an elasticity $\mathscr{M}$-tensor if and only if $\mathscr{A} + t\mathscr{E}$ is a nonsingular elasticity $\mathscr{M}$-tensor for any $t > 0$.
\end{proposition}
\begin{proof}
  Since $\mathscr{A}$ is an elasticity $\mathscr{M}$-tensor, there exists a nonnegative elasticity tensor $\mathscr{B}$ and $s \geq \rho_M(\mathscr{B})$ such that $\mathscr{A} = s\mathscr{E} - \mathscr{B}$.
  Then for any $t>0$, we have $\mathscr{A}+t\mathscr{E}=(s+t)\mathscr{E}-\mathscr{B}$.
  Clearly, $s+t>\rho(\mathscr{B})$, which implies that $\mathscr{A} + t\mathscr{E}$ is a nonsingular elasticity $\mathscr{M}$-tensor.

  Conversely, if $\mathscr{A} + t\mathscr{E}$ is a nonsingular elasticity $\mathscr{M}$-tensor for any $t > 0$, then by the previous theorem we have $\alpha_t > \rho_M \big(\alpha_t \mathscr{E} - (\mathscr{A} + t\mathscr{E})\big)$, where $\alpha_t$ is the greatest diagonal entry of $\mathscr{A} + t \mathscr{E}$.
  Denote $\alpha$ as the greatest diagonal entry of $\mathscr{A}$.
  Then $\alpha_t = \alpha + t$, thus $\alpha + t > \rho_M (\alpha \mathscr{E} - \mathscr{A})$ for any $t>0$.
  When $t$ approaches $0$, it can be concluded that $\alpha \geq \rho_M (\alpha \mathscr{E} - \mathscr{A})$, which implies that $\mathscr{A}$ is an elasticity $\mathscr{M}$-tensor.
\end{proof}

It is well-known that a symmetric nonsingular ${\bf M}$-matrix is positive definite \cite{BermanPlemmons94}.
The same statement was also proved for symmetric nonsingular $\mathscr{M}$-tensors in \cite{ZhangQiZhou14}.
Moreover, we shall show that a nonsingular elasticity $\mathscr{M}$-tensor is M-positive definite thus satisfies the strong ellipticity condition.
In this spirit, we find a class of structured tensors that satisfies the strong ellipticity condition.

\begin{theorem}\label{thm_pd}
  Let $\mathscr{A} \in \mathbb{E}_{4,n}$ be an elasticity $\mathscr{Z}$-tensor. Then $\mathscr{A}$ is a nonsingular elasticity $\mathscr{M}$-tensor if and only if $\mathscr{A}$ is M-positive definite; and $\mathscr{A}$ is an elasticity $\mathscr{M}$-tensor if and only if $\mathscr{A}$ is M-positive semidefinite.
\end{theorem}
\begin{proof}
  Denote $\mathscr{A} = s \mathscr{E} - \mathscr{B}$, where $\mathscr{B}$ is nonnegative.

  If $\mathscr{A}$ is a nonsingular elasticity $\mathscr{M}$-tensor, then $s > \rho_M(\mathscr{B})$.
  By \eqref{eq_maxmin}, we have $s > \mathscr{B} {\bf x}^2 {\bf y}^2$ for all ${\bf x}^\top {\bf x} = {\bf y}^\top {\bf y} = 1$.
  Recall that $\mathscr{E} {\bf x}^2 {\bf y}^2 = ({\bf x}^\top {\bf x}) ({\bf y}^\top {\bf y})$.
  Then $s \mathscr{E} {\bf x}^2 {\bf y}^2 > \mathscr{B} {\bf x}^2 {\bf y}^2$, which is equivalent to $\mathscr{A} {\bf x}^2 {\bf y}^2 > 0$ for all ${\bf x}^\top {\bf x} = {\bf y}^\top {\bf y} = 1$.
  Therefore $\mathscr{A}$ is M-positive definite.

  On the other hand, suppose that $\mathscr{A}$ is M-positive definite, i.e., $\mathscr{A} {\bf x}^2 {\bf y}^2 > 0$ for all ${\bf x}^\top {\bf x} = {\bf y}^\top {\bf y} = 1$.
  Then similarly we have $s = s \mathscr{E} {\bf x}^2 {\bf y}^2 > \mathscr{B} {\bf x}^2 {\bf y}^2$ for all ${\bf x}^\top {\bf x} = {\bf y}^\top {\bf y} = 1$.
  We know from \eqref{eq_maxmin} that $s > \rho_M(\mathscr{B})$, i.e, $\mathscr{A}$ is a nonsingular elasticity $\mathscr{M}$-tensor.

  The result for general elasticity $\mathscr{M}$-tensors can be proved similarly.
\end{proof}

The following equivalent definitions for elasticity $\mathscr{M}$-tensors is straightforward corollary of Proposition \ref{Pro2.2}, Lemma \ref{lem1216} and Theorem \ref{thm_pd}.
\begin{corollary}\label{coro 5.4}
Let $\mathscr{A} \in \mathbb{E}_{4,n}$ be an elasticity $\mathscr{Z}$-tensor.
\begin{enumerate}
  \item $\mathscr{A}$ is an (nonsingular) elasticity $\mathscr{M}$-tensor if and only if
        $$
        \min \big\{ \mathscr{A} {\bf x}^2 {\bf y}^2:\, {\bf x},{\bf y} \in \mathbb{R}_{+}^n,\, {\bf x}^{\top}{\bf x}={\bf y}^{\top}{\bf y}=1 \big\}\geq 0(>0).
        $$
  \item Further assume that $\mathscr{A}$ is irreducible. $\mathscr{A}$ is an (nonsingular) elasticity $\mathscr{M}$-tensor if and only if
        $$
        \min \big\{ \mathscr{A} {\bf x}^2 {\bf y}^2:\, {\bf x},{\bf y} \in \mathbb{R}_{++}^n,\, {\bf x}^{\top}{\bf x}={\bf y}^{\top}{\bf y}=1 \big\}\geq 0(>0).
        $$
\end{enumerate}
\end{corollary}

Recall that the S-eigenvalues of a tensor in $\mathbb{E}_{4,n}$ are defined by the eigenvalues of its unfolding matrices ${\bf A}_x$ and ${\bf A}_y$.
Of course, we can also define $\mathscr{M}$-tensors with respect to S-eigenvalues, which coincides with those tensors $\mathscr{A}$ whose unfolding matrices ${\bf A}_x$ and ${\bf A}_y$ are ${\bf M}$-matrices.
In this case, $\mathscr{A}$ is also M-positive semidefinite since ${\bf A}_x$ and ${\bf A}_y$ are positive semidefinite matrices.
However, the converse may still not hold, when $\mathscr{A}$ is an elasticity $\mathscr{M}$-tensor, as shown by the following example.

\begin{example}
  Consider the case $n=2$. Let $\mathscr{A} \in \mathbb{E}_{4,2}$ be defined by
  $$
  \begin{array}{lll}
    a_{1111}=13, & a_{1122}=2, & a_{2211}=2, \\
    a_{2222}=12, & a_{1112}=-2, & a_{1211}=-2, \\
    a_{1212}=-4, & a_{1222}=-1, & a_{2212}=-1.
  \end{array}
  $$
  By our calculations with Mathematica, $\mathscr{A}$ has six M-eigenvalues: $13.4163$, $12.1118$, $11.2036$, $6.1778$, $0.2442$, and $0.1964$.
  The minimal M-eigenvalue of $\mathscr{A}$ is $0.1964$, which is positive.
  Thus $\mathscr{A}$ is a nonsingular elasticity $\mathscr{M}$-tensor by Theorem \ref{thm_pd}.
  Nonetheless, the unfolding matrices of $\mathscr{A}$ are
  $$
  {\bf A}_x = {\bf A}_y =
  \begin{bmatrix}
    13 & -2 & -2 & -4 \\
    -2 & 2 & -4 & -1 \\
    -2 & -4 & 2 & -1 \\
    -4 & -1 & -1 & 12
  \end{bmatrix},
  $$
  with four eigenvalues: $-2.8331$, $6.0000$, $9.2221$, and $16.6110$.
  There is a negative eigenvalue, which implies that ${\bf A}_x$ and ${\bf A}_y$ are not positive semidefinite and thus not ${\bf M}$-matrices.
\end{example}

We now provide some equivalent definitions of nonsingular elasticity $\mathscr{M}$-tensors, which serve as verification conditions. Recall the definitions of the two $n$-by-$n$ matrices $\mathscr{A} {\bf x}^2$ and $\mathscr{A} {\bf y}^2$ in Section \ref{sec_se}.
The next theorem shows that these two matrices admit the same structures with the original elasticity tensor.

\begin{theorem}\label{thm_Axx}
  Let $\mathscr{A} \in \mathbb{E}_{4,n}$ be an elasticity $\mathscr{Z}$-tensor.
  Then $\mathscr{A}$ is a nonsingular elasticity $\mathscr{M}$-tensor if and only if $\mathscr{A} {\bf x}^2$ is a nonsingular ${\bf M}$-matrix for each ${\bf x} \geq 0$;
  $\mathscr{A}$ is an elasticity $\mathscr{M}$-tensor if and only if $\mathscr{A} {\bf x}^2$ is an ${\bf M}$-matrix for each ${\bf x} \geq 0$.
\end{theorem}
\begin{proof}
  Suppose that $\mathscr{A}$ is a nonsingular elasticity $\mathscr{M}$-tensor.
  Then we know by \eqref{eq_Axx} that $\mathscr{A} {\bf x}^2 $ is positive definite for each ${\bf x} \in \mathbb{R}^n$ since $\mathscr{A}$ is M-positive definite.
  Another simple observation is that $\mathscr{A} {\bf x}^2 $ is a ${\bf Z}$-matrix for each ${\bf x} \geq 0$ when $\mathscr{A}$ is an elasticity $\mathscr{Z}$-tensor.
  Thus $\mathscr{A} {\bf x}^2 $ is a positive definite ${\bf Z}$-matrix for each ${\bf x} \geq 0$.
  From the equivalent definitions of nonsingular ${\bf M}$-matrices \cite{BermanPlemmons94}, it can be concluded that $\mathscr{A} {\bf x}^2 $ is a nonsingular ${\bf M}$-matrix for each ${\bf x} \geq 0$.

  Conversely, if $\mathscr{A} {\bf x}^2 $ is a nonsingular ${\bf M}$-matrix for each ${\bf x} \geq 0$, then $\mathscr{A} {\bf x}^2 $ is always positive definite.
  That is, $\mathscr{A} {\bf x}^2 {\bf y}^2 = {\bf y}^\top \mathscr{A} {\bf x}^2  {\bf y} > 0$ for each ${\bf x} \geq 0$ and ${\bf y} \in \mathbb{R}^n$.
  Write $\mathscr{A} = s \mathscr{E} - \mathscr{B}$, where $\mathscr{B}$ is nonnegative.
  Then $s > \mathscr{B} {\bf x}^2 {\bf y}^2$ for each ${\bf x},{\bf y} \geq 0$ satisfying ${\bf x}^\top {\bf x} = {\bf y}^\top {\bf y} = 1$.
  Hence, Corollary \ref{coro nn} tells that $s > \rho_M(\mathscr{B})$, i.e., $\mathscr{A}$ is a nonsingular elasticity $\mathscr{M}$-tensor.
\end{proof}

Similarly, we have a parallel result for $\mathscr{A}{\bf y}^2$.

\begin{theorem}\label{thm_Ayy}
  Let $\mathscr{A} \in \mathbb{E}_{4,n}$ be an elasticity $\mathscr{Z}$-tensor. Then $\mathscr{A}$ is a nonsingular elasticity $\mathscr{M}$-tensor if and only if $\mathscr{A}{\bf y}^2$ is a nonsingular ${\bf M}$-matrix for each ${\bf y} \geq 0$; $\mathscr{A}$ is an elasticity $\mathscr{M}$-tensor if and only if $\mathscr{A}{\bf y}^2$ is an ${\bf M}$-matrix for each ${\bf y} \geq 0$.
\end{theorem}

There is a well-known equivalent definition for nonsingular ${\bf M}$-matrices called semi-positivity.
That is, a ${\bf Z}$-matrix ${\bf A}$ is a nonsingular ${\bf M}$-matrix if and only if there exits a positive (or equivalently nonnegative) vector ${\bf x}$ such that ${\bf A} {\bf x}$ is also a positive vector.
Ding, Qi, and Wei \cite{DingQiWei13} proved that this also holds for nonsingular $\mathscr{M}$-tensors.
The semi-positivity is essential to verify whether a tensor is a nonsingular $\mathscr{M}$-tensor and is also important for solving the polynomial systems of equations with $\mathscr{M}$-tensors \cite{DingWei16}.
Combining the semi-positivity of nonsingular ${\bf M}$-matrices and Theorems \ref{thm_Axx} and \ref{thm_Ayy}, we have the following equivalent definitions for nonsingular elasticity $\mathscr{M}$-tensors immediately.

\begin{theorem}\label{thm_Axxy}
  Let $\mathscr{A} \in \mathbb{E}_{4,n}$ be an elasticity $\mathscr{Z}$-tensor. The following conditions are equivalent:
  \begin{enumerate}[{\rm (1)}]
    \item $\mathscr{A}$ is a nonsingular elasticity $\mathscr{M}$-tensor;
    \item For each ${\bf x} \geq 0$, there exists ${\bf y} > 0$ such that $\mathscr{A} {\bf x}^2 {\bf y} > 0$;
    \item For each ${\bf x} \geq 0$, there exists ${\bf y} \geq 0$ such that $\mathscr{A} {\bf x}^2 {\bf y} > 0$;
    \item For each ${\bf y} \geq 0$, there exists ${\bf x} > 0$ such that $\mathscr{A} {\bf x} {\bf y}^2 > 0$;
    \item For each ${\bf y} \geq 0$, there exists ${\bf x} \geq 0$ such that $\mathscr{A} {\bf x} {\bf y}^2 > 0$.
  \end{enumerate}
\end{theorem}

A matrix ${\bf A} \in \mathbb{R}^{n \times n}$ is called strictly diagonally dominant if
$$
|a_{ii}| > \sum_{j \neq i} |a_{ij}|, \quad i = 1,2,\dots,n.
$$
Condition (2) in Theorem \ref{thm_Axxy} states that for each nonnegative vector ${\bf x}$, there exists a positive vector ${\bf y}$ such that $\mathscr{A} {\bf x}^2  {\bf y} = \mathscr{A} {\bf x}^2 {\bf y} > 0$.
Denote a diagonal matrix ${\bf D}$ with $d_{ii} = y_i$ for $i=1,2,\dots,n$ and $\widetilde{\bf A} := (\mathscr{A} {\bf x}^2 ) {\bf D}$.
When $\mathscr{A}$ be an elasticity $\mathscr{Z}$-tensor, the matrix $\widetilde{\bf A}$ is also a ${\bf Z}$-matrix.
Thus we have
$$
|\widetilde{a}_{ii}| - \sum_{j \neq i} |\widetilde{a}_{ij}|
= \widetilde{a}_{ii} + \sum_{j \neq i} \widetilde{a}_{ij}
= (\mathscr{A} {\bf x}^2 {\bf y})_i > 0, \quad i = 1,2,\dots,n,
$$
which implies that $\widetilde{\bf A}$ is strictly diagonally dominant.
Applying the above discussion, we can prove the following corollary of Theorem \ref{thm_Axxy}.

\begin{corollary}\label{coro_dd}
  Let $\mathscr{A} \in \mathbb{E}_{4,n}$ be an elasticity $\mathscr{Z}$-tensor. The following conditions are equivalent:
  \begin{enumerate}[{\rm (1)}]
    \item $\mathscr{A}$ is a nonsingular elasticity $\mathscr{M}$-tensor;
    \item For each ${\bf x} \geq 0$, there exists a positive diagonal matrix ${\bf D}$ such that ${\bf D} (\mathscr{A} {\bf x}^2)  {\bf D}$ is strictly diagonally dominant;
    \item For each ${\bf y} \geq 0$, there exists a positive diagonal matrix ${\bf D}$ such that ${\bf D} (\mathscr{A}{\bf y}^2) {\bf D}$ is strictly diagonally dominant.
  \end{enumerate}
\end{corollary}


\section{Conclusions}

We have established several sufficient conditions for the strong ellipticity (M-positive definiteness) of general elasticity tensors. Our first sufficient condition extends the coverage of S-PSD tensors, which states that $\mathscr{A}$ is M-PSD or M-PD if it can be modified into an S-PSD or S-PD tensor $\mathscr{B}$ respectively by preserving $b_{ijkl} = b_{jilk}$ and $b_{ijkl} + b_{jikl} = a_{ijkl} + a_{jikl}$.
To check whether a tensor satisfies this condition, we employ an alternating projection method called POCS and display its convergence.

Next, we consider the properties for nonnegative elasticity tensors. A Perron-Frobenius type theorem for M-spectral radii of a nonnegative elasticity tensor has been proposed in Section \ref{sec_nonneg}. Then we investigate a class of tensors satisfying the SE-condition, the elasticity $\mathscr{M}$-tensor.
Combining Theorems \ref{thm_mindiag}, \ref{thm_pd},  \ref{thm_Axx}-- \ref{thm_Axxy} and Corollary \ref{coro 5.4}, \ref{coro_dd}, we summarize the equivalent definitions for nonsingular elasticity $\mathscr{M}$-tensors given in this paper.
Let $\mathscr{A} \in \mathbb{E}_{4,n}$ be an elasticity $\mathscr{Z}$-tensor. The following conditions are equivalent:
  \begin{enumerate}[{\rm (C1)}]
    \item $\mathscr{A}$ is a nonsingular elasticity $\mathscr{M}$-tensor;
    \item $\mathscr{A}$ is M-positive definite, i.e., $\mathscr{A} {\bf x}^2 {\bf y}^2 > 0$ for all nonzero ${\bf x}, {\bf y} \in \mathbb{R}^n$;
    \item $\min \big\{ \mathscr{A} {\bf x}^2 {\bf y}^2:\, {\bf x},{\bf y} \in \mathbb{R}_{+}^n,\, {\bf x}^{\top}{\bf x}={\bf y}^{\top}{\bf y}=1 \big\}> 0$;
    \item All the M-eigenvalues of $\mathscr{A}$ are positive;
    \item $\alpha > \rho_M (\alpha \mathscr{E} - \mathscr{A})$, where $\alpha = \max \big\{ a_{iikk}:\, i,k=1,2,\dots,n \big\}$;
    \item For each ${\bf x} \geq 0$, $\mathscr{A} {\bf x}^2 $ is a nonsingular ${\bf M}$-matrix;
    \item For each ${\bf x} \geq 0$, there exists ${\bf y} > 0$ such that $\mathscr{A} {\bf x}^2 {\bf y} > 0$;
    \item For each ${\bf x} \geq 0$, there exists ${\bf y} \geq 0$ such that $\mathscr{A} {\bf x}^2 {\bf y} > 0$;
    \item For each ${\bf x} \geq 0$, there exists a positive diagonal matrix ${\bf D}$ such that ${\bf D} (\mathscr{A} {\bf x}^2 ) {\bf D}$ is strictly diagonally dominant;
    \item For each ${\bf y} \geq 0$, $\mathscr{A}{\bf y}^2$ is a nonsingular ${\bf M}$-matrix;
    \item For each ${\bf y} \geq 0$, there exists ${\bf x} > 0$ such that $\mathscr{A} {\bf x} {\bf y}^2 > 0$;
    \item For each ${\bf y} \geq 0$, there exists ${\bf x} \geq 0$ such that $\mathscr{A} {\bf x} {\bf y}^2 > 0$;
    \item For each ${\bf y} \geq 0$, there exists a positive diagonal matrix ${\bf D}$ such that ${\bf D} (\mathscr{A}{\bf y}^2 ){\bf D}$ is strictly diagonally dominant.
  \end{enumerate}


\end{document}